\documentclass[10pt]{article} 
\usepackage{latexsym, amsmath}
\usepackage{lscape}
\usepackage{epstopdf,caption,subcaption,graphicx,xcolor}
\usepackage{tikz}
\usepackage{pgfgantt}
\usepackage{hyperref}
\usetikzlibrary{shapes,arrows,fit,calc,positioning}
\usetikzlibrary{decorations.pathreplacing}
\normalsize
\newtheorem{theorem}{Theorem}
\newtheorem{lemma}{Lemma}

\newtheorem{definition}{Definition}

\newenvironment{proof}{{\sc Proof. }}{\hfill$\Box$\vspace{0.1in}}
\def\mcH{\mathcal{H}}
\def\mcK{\mathcal{K}}
\def\mcP{\mathcal{P}}
\def\mcQ{\mathcal{Q}}
\title{Approximation algorithms for non-sequential star packing problems}
\author{
Mengyuan~Hu\thanks{Department of Mathematics, Hangzhou Dianzi University. Hangzhou 310018, China.
	Email: \texttt{\{2341070105,anzhang, chenyong\}@hdu.edu.cn}} 
\and
An~Zhang$^*$$^\ddagger$
\and
Yong~Chen$^*$
\and
Mingyang~Gong\thanks{Department of Computing Science, University of Alberta.  Edmonton, Alberta T6G 2E8, Canada.
	Email: \texttt{\{mgong4, guohui\}@ualberta.ca}}
\and
Guohui~Lin$^\dagger$\thanks{Corresponding authors.}
}
\date{}
\begin{document}
\maketitle
\begin{abstract}
For a positive integer $k \ge 1$, a $k$-star ($k^+$-star, $k^-$-star, respectively) is a connected graph containing a degree-$\ell$ vertex and $\ell$ degree-$1$ vertices,
where $\ell = k$ ($\ell \ge k$, $1 \le \ell \le k$, respectively).
The $k^+$-star packing problem is to cover as many vertices of an input graph $G$ as possible using vertex-disjoint $k^+$-stars in $G$;
and given $k > t \ge 1$, the $k^-/t$-star packing problem is to cover as many vertices of $G$ as possible using vertex-disjoint $k^-$-stars but no $t$-stars in $G$.
Both problems are NP-hard for any fixed $k \ge 2$.
We present a $(1 + \frac {k^2}{2k+1})$- and a $\frac 32$-approximation algorithms for the $k^+$-star packing problem when $k \ge 3$ and $k = 2$, respectively,
and a $(1 + \frac 1{t + 1 + 1/k})$-approximation algorithm for the $k^-/t$-star packing problem when $k > t \ge 2$.
They are all local search algorithms and they improve the best known approximation algorithms for the problems, respectively.

\end{abstract}

\section*{Acknowledgments}
This research is supported by the NSERC Canada,
the National Natural Science Foundation of China (Grants 12371316, 12471301),
the Zhejiang Provincial Natural Science Foundation (Grant LY21A010014),
and the Ministry of Science and Technology of China (Grants G2022040004L, G2023016016L).

\section{Introduction}\label{sec:intro}
For a positive integer $k \ge 1$, a $k$-star ($k^+$-star, $k^-$-star, respectively) is a connected graph containing a degree-$\ell$ vertex, called the {\em center},
and $\ell$ degree-$1$ vertices, called {\em satellites}, where $\ell = k$ ($\ell \ge k$, $1 \le \ell \le k$, respectively).

We study the problem to cover the maximum number of vertices in an input graph using vertex-disjoint stars, referred to as {\em star packing}.
Such a problem, and many variants, have a rich literature due to their theoretical importance and numerous networking applications such as wireless sensor network~\cite{LCC17}.
Specially, we study two {\em non-sequential} star packing problems from the approximation algorithm perspective.
For one of them, the candidate stars each has at least $k$ satellites for a given constant $k \ge 2$, called the {\em $k^+$-star packing} problem,
and we propose a $(1 + \frac {k^2}{2k+1})$- and a $\frac 32$-approximation algorithms when $k \ge 3$ and $k = 2$, respectively;
for the other, the candidate stars each has at most $k$ but not exactly $t$ satellites where $k > t \ge 2$, called the {\em $k^-/t$-star packing} problem,
and we propose a $(1 + \frac 1{t + 1 + 1/k})$-approximation algorithm.
All of them improve the previous best results, respectively.

The studied problems fall in the general $\mcH$-packing problem.
Given a collection $\mcH$ of (non-isomorphic) graphs,
an {\em $\mcH$-packing} of an input graph $G = (V, E)$ is a set of {\em vertex-disjoint} subgraphs of $G$ in which each subgraph is isomorphic to a graph of $\mcH$.
A vertex of $G$ is {\em covered} by the $\mcH$-packing if it appears in a subgraph inside the packing.
The goal of the $\mcH$-packing problem is to find an $\mcH$-packing that covers the maximum number of vertices.
In particular, a {\em perfect} $\mcH$-packing or a so-called {\em $\mcH$-factor} refers to an $\mcH$-packing that covers all the vertices in the input graph.
The existence of an $\mcH$-factor depends on $\mcH$ and the input graph;
$\mcH$ and sometimes the class of the input graphs determine the complexity and approximability of the optimization $\mcH$-packing problem.

We first review some closely related $\mcH$-packing problems, and then the star packing problems for which $\mcH$ consists of only stars.

For a positive integer $i$, $K_i$ ($C_i$, $P_i$, respectively) denotes the complete graph (the single cycle, the single path, respectively) on $i$ vertices.
One sees that when $\mcH = \{K_2\}$, the $\mcH$-packing problem is exactly the classical {\em maximum matching} problem,
which is solvable in polynomial time~\cite{GK04}. 
When $\mcH = \{K_3\}$, the $\mcH$-packing problem is the classical {\em triangle packing} problem,
which is polynomially solvable on subcubic graphs (i.e., graphs of maximum degree at most $3$)~\cite{CR02} but becomes APX-hard on graphs of maximum degree $4$ or larger~\cite{Kan91}.
When $\mcH = \{P_3\}$, the $\mcH$-packing problem or the {\em maximum $P_3$-packing} problem is NP-hard on subcubic graphs~\cite{MT07},
and NP-hard on claw-free cubic planar graphs~\cite{XL21}.
More generally, Kirkpatrick and Hell~\cite{KH83} showed that, when $\mcH$ contains a single graph $H$,
the $\mcH$-packing problem is NP-hard as long as $H$ is connected and contains at least three vertices;
and this NP-hardness result holds even when restricted to planar graphs~\cite{BJL90}.

When $\mcH = \{K_2, H\}$ for some $H$ non-isomorphic to $K_2$,
Loebl and Poljak~\cite{LP93} proved that the $\mcH$-packing problem is NP-hard unless $H$ is
perfect matchable (i.e., $H$ admits a perfect matching),
or hypomatchable (i.e., the remainder graph $H - v$ of $H$ after removing the vertex $v$ admits a perfect matching, for any vertex $v$ in $H$),
or a propeller (i.e., $H$ can be obtained from a hypomatchable graph by adding a new pair of vertices $u$ and $v$, a new edge $(u, v)$,
	and some new edges connecting $u$ to at least one vertex in the hypomatchable graph).
When $\mcH$ consists of $K_2$ and some other non-isomorphic graphs $H_1, H_2, \ldots, H_\ell$,
then the $\mcH$-packing problem is polynomially solvable if every $H_i$ is hypomatchable~\cite{HK84}.

When $\mcH$ consists of only complete graphs, the $\mcH$-packing problem is NP-hard unless $K_2\in \mcH$~\cite{HK84}.
When $\mcH$ contains only cycles, Hell et al.~\cite{HKK88} showed that the $\mcH$-packing problem is NP-hard unless
$\mcH$ contains all cycles, or all cycles but $C_3$, or all cycles but $C_4$.
When $\mcH$ is composed of all paths of order $k$ or above, the $\mcH$-packing problem or the $k^+$-path packing problem is polynomially solvable for $k \le 3$~\cite{CCL22}
while becomes NP-hard for $k \ge 4$~\cite{KLM23}.

Closely related to the $\mcH$-packing problem,
in the {\em weighted set packing} problem one is given a ground set and a collection of subsets of the ground set each associated with a non-negative weight,
and the goal is to find a sub-collection of pairwise disjoint sets of the maximum total weight.
The weighted set packing problem is a central optimization problem~\cite{GJ79} that has received numerous studies from the approximation algorithm perspective,
and the state-of-the-art includes a $(\frac k2 + \epsilon)$-approximation algorithm~\cite{Neu23}, where $k$ is the maximum cardinality of the given sets and $\epsilon > 0$,
a $1.786$-approximation algorithm when $k = 3$~\cite{TW23},
and a $(\frac {k+1}3 + \epsilon)$-approximation algorithm when the sets are uni-weighted~\cite{CGM13,FY14},
all of which are based on local improvement.

Since the graph isomorphism problem is still not polynomially solvable~\cite{HBD17},
one cannot afford to reduce the $\mcH$-packing problem to the weighted set packing problem if the graphs in $\mcH$ are arbitrary,
but otherwise such as the triangle packing and the $P_3$-packing problems can be reduced to the uni-weighted $3$-set packing problem
to obtain their first $(\frac 43 + \epsilon)$-approximation algorithms.
Nevertheless, typically directly designing local operations for the $\mcH$-packing problem leads to better performance analysis.
To name a few, the triangle packing problem on graphs of maximum degree $4$ admits a $\frac 65$-approximation algorithm~\cite{MW08},
and the $P_3$-packing problem on cubic graphs admits a $\frac 43$-approximation algorithm~\cite{KM04}.
For the $k^+$-path packing problem, Gong et al.~\cite{GFL22} proposed a local improvement algorithm which achieves a ratio of $\rho(k) \leq 0.4394k + 0.6576$ for any fixed $k \ge 4$;
when $k = 4$, Kobayashi et al.~\cite{KLM23} presented a simple local improvement algorithm which turns out to be a $4$-approximation,
Gong et al.~\cite{GFL22} proposed a $2$-approximation algorithm using more local operations;
and most recently Gong et al.~\cite{GCL23} proposed to begin with a maximum matching and then apply multiple local improvement operations, resulting in a $1.874$-approximation algorithm.

In this paper, we study the star packing problem, i.e., the $\mcH$-packing problem where $\mcH$ is a collection of various stars.
In the literature, an $i$-star is also written as $K_{1, i}$, that is, the complete bipartite graph with a singleton on one side and $i$ vertices on the other side.
In particular, the $k^+$-star packing problem is the $\mcH$-packing problem where $\mcH = \{K_{1, k}, K_{1, k+1}, K_{1, k+2}, \ldots\}$, for a given $k \ge 1$,
and the $k^-/t$-star packing problem is the one where $\mcH = \{K_{1, 1}, \ldots, K_{1, t-1}, K_{1, t+1}, \ldots, K_{1, k}\}$, for given $k > t \ge 1$.
Hell and Kirkpatrick~\cite{HK86} showed that the star packing problem is polynomially solvable if
$\mcH$ is composed of {\em sequential} stars, including the cases where $\mcH$ contains all stars, or $\mcH$ contains all the stars with up to $k$ satellites for any $k \ge 1$;
otherwise, the same authors showed that the problem is strongly NP-hard~\cite{HK84}.
Furthermore, the $k^+$-star packing problem for any $k \ge 2$ is NP-hard even on bipartite graphs~\cite{LL21},
and the $2^+$-star packing problem is NP-hard even on bipartite graphs with maximum degree $4$~\cite{LL21} or on cubic graphs~\cite{XLL24}.
On the positive side, the $2^+$-star packing problem admits a $2$-approximation algorithm~\cite{LL21};
and on cubic graphs it admits a better $\frac 76$-approximation algorithm~\cite{XLL24}.
The $2$-approximation is recently improved to a $\frac 95$-approximation algorithm by Huang et al.~\cite{HZG24},
who also presented a local search $(1 + \frac k2)$-approximation algorithm for the $k^+$-star packing problem for any $k \ge 3$.
For the $k^-/t$-star packing problem with $k > t \ge 2$, Li and Lin~\cite{LL21} observed a simple $(1 + \frac 1t)$-approximation algorithm.
There is not much result on inapproximability, except the APX-completeness of the $3$-star packing on cubic graphs,
which is equivalent to the maximum distance-$3$ independent set problem~\cite{EIL17,XL24}.

We continue to study the $k^+$-star packing problem with $k \ge 2$ and the $k^-/t$-star packing problem with $k > t \ge 2$ from the perspective of approximation algorithms.
Our algorithms are all based on local search or local improvement.
For the first problem, we design operations to increase the number of vertices covered by the current feasible $k^+$-star packing;
for the second problem, the operations are designed to decrease the number of $t$-stars in the current optimal sequential $k^-$-star packing,
or to keep the number of $t$-stars but increase the number of other stars.
These algorithms are presented in Sections~2 and~3, respectively.
Table~\ref{tab01} below summarizes the previous approximation results and the improved ones achieved in this paper.
We conclude the paper in Section 4.
%
\begin{table}[h]
\centering
\captionsetup{width=.95\textwidth}
\caption{Previous approximation results and the improved ones achieved for the two non-sequential star packing problems.\label{tab01}}
\begin{tabular}{l||l}
  Problem 					& Approximation ratios \\
  \hline
  \hline
  $k^+$-star packing with $k \ge 3$				& $1 + \frac k2$~\cite{HZG24} $\longrightarrow$ $1 + \frac {k^2}{2k+1}$ (Thm~\ref{thm01})\\
  \hline
  $2^+$-star packing 									& $2$~\cite{LL21} $\longrightarrow$ $\frac 95$~\cite{HZG24} $\longrightarrow$ $\frac 32$ (Thm~\ref{thm02})\\
  \hline
  $k^-/t$-star packing with $k > t \ge 2$	& $1 + \frac 1t$~\cite{LL21} $\longrightarrow$ $1 + \frac 1{t + 1 + 1/k}$ (Thm~\ref{thm03})\\
  \hline
  $k^-/t$-star packing with $k = \infty$ and $t \ge 2$	& $1 + \frac 1{t + 2}$ (Extended from Thm~\ref{thm03})
\end{tabular}
\end{table}

\section{Approximating the $k^+$-star packing problem}
The previous best results for the $k^+$-star packing problem when $k \ge 3$ is the $(1 + \frac k2)$-approximation algorithm,
and when $k = 2$ is the $\frac 95$-approximation algorithm, both by Huang et al.~\cite{HZG24} (see Table~\ref{tab01}).
While the former is a relatively simple local search algorithm,
the latter $\frac 95$-approximation algorithm is involved and,
it first computes a maximal collection of trees, 
then connects the other vertices to these trees,
and lastly decomposes the resulting trees carefully into a solution. 

We first present our algorithm for any $k \ge 2$, which is a local search algorithm denoted as {\sc LocalSearch-$k^+$}, in the next subsection,
and then its performance analysis in Subsection 2.2.
When $k = 2$, the general algorithm {\sc LocalSearch-$k^+$} works out to be a $\frac 95$-approximation;
we then add one more operation to the general algorithm to become the final {\sc LocalSearch-$2^+$}, and show that it is a $\frac 32$-approximation.

We fix an input graph $G = (V, E)$ for presentation;
a feasible solution $\mcP$ is a collection of vertex-disjoint $k^+$-stars in $G$, and a star in $\mcP$ is called an {\em internal} star.
We use $V(\mcP)$ to denote the set of the vertices in the internal stars, i.e., the covered vertices by $\mcP$.
A vertex in $V \setminus V(\mcP)$ is {\em uncovered},
and the subgraph $R$ induced by all the uncovered vertices is the {\em remainder} graph with respect to $\mcP$.
A star in $R$ is called an {\em outside} star (with respect to $\mcP$ --- which is not repeated in the sequel).

\subsection{The algorithm for $k \ge 2$}
Similar to many other local search algorithms such as in Huang et al.~\cite{HZG24},
we define several operations to be repeatedly executed inside our algorithm {\sc LocalSearch-$k^+$}.
Note that the algorithm is iterative, and during each iteration a feasible solution is assumed at the beginning, referred to as the {\em current} solution or packing,
an operation is applied and the solution is {\em updated}, and then the iteration ends.
When none of the operations is applicable, the algorithm terminates and returns the current solution as the final solution.

\begin{definition}
\label{def01}
{\em (Operation {\sc Collect($v$)})}
Given a vertex $v$ in the current remainder graph,
the operation {\sc Collect($v$)} extracts an $\ell$-star centered at $v$ from the remainder graph, where $\ell \ge k$ is the degree of $v$ in the remainder graph,
and adds it to the current solution $\mcP$.
\end{definition}

The algorithm starts with the empty solution,
and in the first phase (or the first a few iterations), {\sc Collect($v$)} is repeatedly applied on a vertex $v$ of degree at least $k$ in the remainder graph until impossible.
Note that at this moment, the maximum degree of the remainder graph is at most $k-1$ and no uncovered vertex is adjacent to the center of any internal star.
This property is stated in the following Lemma~\ref{lemma01}.
In fact, at the end of each iteration during the algorithm,
{\sc Collect($v$)} is {\em re-applied} to the center $v$ of every internal $k^+$-star to pick up the uncovered vertices adjacent to $v$, if any.

\begin{lemma}
\label{lemma01}
When operation {\sc Collect} is not applicable,
the maximum degree of the remainder graph with respect to the feasible solution $\mcP$ is at most $k-1$ and no uncovered vertex is adjacent to the center of any internal star.
\end{lemma}

\begin{definition}
\label{def02}
{\em (Operation {\sc Pull-by-$(k+1)^+$})}
Given a satellite $v$ of an internal $(k+1)^+$-star, 
the operation removes $v$ from the internal $(k+1)^+$-star and applies {\sc Collect} to extract a $k^+$-star.
\end{definition}

Note that such an operation {\sc Pull-by-$(k+1)^+$} trades a satellite $v$ of an internal $(k+1)^+$-star for a $k^+$-star,
which covers $v$ since prior to this operation no {\sc Collect} is applicable.
There are two possible scenarios:
In one scenario $v$ is adjacent to at least $k$ uncovered vertices and thus $v$ is the center of the extracted star;
in the other scenario $v$ is adjacent to the center of an outside $(k-1)$-star and thus $v$ is a satellite of the extracted star.
We remark that to apply {\sc Pull-by-$(k+1)^+$}, only one satellite of an internal $(k+1)^+$-star is examined, disregarding how large the internal star is.
For example, even if there are two satellites of an internal $(k+2)^+$-star adjacent to the center of an outside $(k-2)$-star,
no {\sc Pull-by-$(k+1)^+$} is applied.

\begin{definition}
\label{def03}
{\em (Operation {\sc Pull-by-$k$})}
Given an internal $k$-star, 
the operation removes the star from $\mcP$ and applies {\sc Collect} to extract a $(k+1)^+$-star, or if not possible then to extract two vertex-disjoint $k$-stars.
\end{definition}

Note that such an operation {\sc Pull-by-$k$} trades an internal $k$-star for a larger $(k+1)^+$-star, or for two vertex-disjoint $k$-stars.
We remark that right after the operation, some vertices in the remainder graph might have their degree greater than or equal to $k$,
and thus {\sc Collect} operations might be applicable to extract more $k^+$-stars.
There are many possible scenarios to apply {\sc Pull-by-$k$}, for example,
when a satellite $v$ of the internal $k$-star is adjacent to $k$ or more uncovered vertices,
or when two satellites of the internal $k$-star are both adjacent to the center of an outside $(k-1)$-star,
or when two satellites of the internal $k$-star each is adjacent to the center of one of the two vertex-disjoint outside $(k-1)$-stars, respectively;
and so on.

After applying the {\sc Pull-by-$(k+1)^+$} and {\sc Pull-by-$k$} operations, we have the following properties stated in Lemma~\ref{lemma02} for the achieved solution $\mcP$.

\begin{lemma}
\label{lemma02} 
When {\sc Collect}, {\sc Pull-by-$(k+1)^+$} and {\sc Pull-by-$k$} are not applicable,
\begin{itemize}
\parskip=0pt
\item
	a satellite of any internal star is adjacent to at most $k-1$ uncovered vertices;
\item
	no satellite of any internal $(k+1)^+$-star is adjacent to the center of an outside $(k-1)$-star;
\item
	an internal $k$-star $S$ has at most one satellite that is adjacent to the center of an outside $(k-1)$-star $T$, and if this happens,
	then no other satellite of $S$ can be adjacent to $k-1$ uncovered vertices excluded from $T$ and
	no two other satellites of $S$ can both be adjacent to the center of an outside $(k-2)$-star vertex-disjoint from $T$.
\end{itemize}
\end{lemma}

We next define two more operations that work on two internal stars to increase the covered vertices.
They can be deemed as combinations of the above two {\sc Pull} operations.
For convenience, a {\em $k_{v\mbox{-}u}$-star} is an internal $k$-star of which the satellite $v$ is adjacent to the center $u$ of an outside $(k-1)$-star,
and a {\em $k^+_v$-star} is an internal $k^+$-star of which the satellite $v$ is adjacent to exactly $k-1$ uncovered vertices
(see for an illustration in Figure~\ref{fig01}).

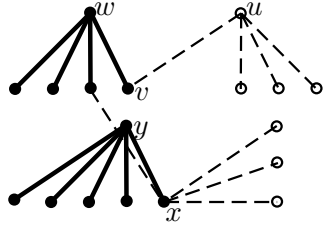
\begin{figure}[ht]
\centering\scalebox{1.0}{
  \setlength{\unitlength}{1bp}%
  \begin{picture}(128.95, 100.35)(0,0)
  \put(0,0){\includegraphics{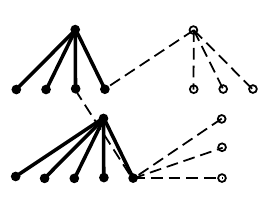}}
  \put(52.83,53.90){\fontsize{11.38}{13.66}\selectfont $v$}
  \put(94.96,85.77){\fontsize{11.38}{13.66}\selectfont $u$}
  \put(64.55,8.12){\fontsize{11.38}{13.66}\selectfont $x$}
  \put(37.60,85.79){\fontsize{11.38}{13.66}\selectfont $w$}
  \put(52.39,40.98){\fontsize{11.38}{13.66}\selectfont $y$}
  \end{picture}%
}
\captionsetup{width=.95\textwidth}
\caption{An illustration to apply operation {\sc Pull-by-$(k, (k+1)^+)$} on two internal $k$-star centered at $w$, which is a $k_{v\mbox{-}u}$-star,
	and $(k+1)^+$-star centered at $y$, which is a $(k+1)^+_x$-star.
	The filled vertices are covered, the empty vertices are uncovered, the edges in the internal stars are solid while the dashed edges are in the input graph $G$.
	In this case, $k = 4$, a vertex of the $k_{v\mbox{-}u}$-star other than $v$ is adjacent to $x$.
	After removing the $k_{v\mbox{-}u}$-star from $\mcP$ and the satellite $x$ from the $(k+1)^+_x$-star,
	two vertex-disjoint $k$-stars centered at $u$ and $x$, respectively, are extracted.\label{fig01}}
\end{figure}

\begin{definition}
\label{def04}
{\em (Operation {\sc Pull-by-$(k, (k+1)^+)$})}
Given two internal $k$-star and $(k+1)^+$-star,
the operation removes the $k$-star from $\mcP$ and a satellite from the $(k+1)^+$-star, and applies {\sc Collect} operations to extract two vertex-disjoint $k^+$-stars
(see for an illustration in Figure~\ref{fig01}).
\end{definition}

\begin{figure}[ht]
\centering\scalebox{1.0}{
  \setlength{\unitlength}{1bp}%
  \begin{picture}(128.64, 100.35)(0,0)
  \put(0,0){\includegraphics{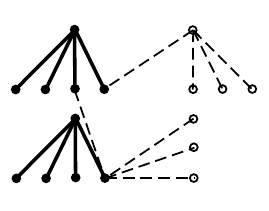}}
  \put(52.52,53.90){\fontsize{11.38}{13.66}\selectfont $v$}
  \put(94.66,85.77){\fontsize{11.38}{13.66}\selectfont $u$}
  \put(51.01,8.12){\fontsize{11.38}{13.66}\selectfont $x$}
  \put(37.29,85.79){\fontsize{11.38}{13.66}\selectfont $w$}
  \put(23.39,40.98){\fontsize{11.38}{13.66}\selectfont $y$}
  \end{picture}%
}
\captionsetup{width=.95\textwidth}
\caption{An illustration to apply operation {\sc Pull-by-$(k, k)$} on two internal $k$-stars centered at $w$ and $y$, which are a $k_{v\mbox{-}u}$-star and a $k_x$-star, respectively.
	The filled vertices are covered, the empty vertices are uncovered, the edges in the internal stars are solid while the dashed edges are in the input graph $G$.
	In this case, $k = 4$, a vertex of the $k_{v\mbox{-}u}$-star other than $v$ is adjacent to $x$.
	After removing the two $k$-stars from $\mcP$, two vertex-disjoint $k$-star centered at $u$ and $(k+1)$-star centered at $x$, respectively, are extracted.\label{fig02}}
\end{figure}
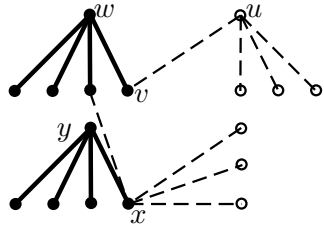

\begin{definition}
\label{def05}
{\em (Operation {\sc Pull-by-$(k, k)$})}
Given two internal $k$-stars,
the operation removes both of them from $\mcP$, and applies {\sc Collect} operations to extract either two vertex-disjoint $k^+$-star and $(k+1)^+$-star,
or three vertex-disjoint $k$-stars
(see for an illustration in Figure~\ref{fig02}).
\end{definition}

Below, we sometimes use simply {\sc Pull} to mean any one of the above four {\sc Pull-by-$(k+1)^+$}, {\sc Pull-by-$k$},
{\sc Pull-by-$(k, (k+1)^+)$}, and {\sc Pull-by-$(k, k)$} operations.
Using {\sc Collect} and {\sc Pull} operations, our algorithm {\sc LocalSearch-$k^+$} repeatedly applies one of them;
when none of them is applicable, the algorithm outputs the current solution $\mcP$ to the $k^+$-star packing problem.
A high-level description of the algorithm is depicted in Figure~\ref{fig03}.

\begin{figure}[h]
\begin{center}
\framebox{
\begin{minipage}{0.9\textwidth}
{\sc LocalSearch-$k^+$} for $k^+$-star packing:\\
Input: A connected graph $G = (V, E)$;\\
Output: A collection $\mcP$ of vertex-disjoint $k^+$-stars in $G$.
\begin{enumerate}
\parskip=0pt
\item
	Initialize $\mcP = \emptyset$ and the remainder graph $R = G$;
\item
	While (a {\sc Collect} or {\sc Pull} operation is applicable) do
	\begin{enumerate}
	\parskip=0pt
	\item
		apply the operation to update $\mcP$;
	\item
		re-apply {\sc Collect} on the center of every internal star;
	\item
		update the remainder graph $R = G[V \setminus V(\mcP)]$;
	\end{enumerate}
\end{enumerate}
\end{minipage}}
\end{center}
\captionsetup{width=.99\textwidth} 
\caption{A high-level description of {\sc LocalSearch-$k^+$} for $k^+$-star packing.\label{fig03}}
\end{figure}

\subsection{Performance analysis}
We fix an optimal $k^+$-star packing, denoted as $\mcQ^*$, for discussion, and denote by $Opt = |V(\mcQ^*)|$ the number of vertices covered by $\mcQ^*$.
For ease of presentation, the stars inside $\mcQ^*$ are referred to as {\em optimal stars}.
We start with some observations on $\mcQ^*$.

\begin{lemma}
\label{lemma03}
At most $k-1$ satellites of an optimal star are uncovered by the computed solution $\mcP$.
\end{lemma}
\begin{proof}
We prove by contradiction to let $u$ denote the center of this optimal star.
If $u$ is uncovered, then {\sc Collect($u$)} is applicable to add a $k^+$-star to the computed solution $\mcP$.
If $u$ is the center of an internal star of $\mcP$, then {\sc Collect($u$)} is re-applicable to add more satellites to the internal star.
If $u$ is a satellite of an internal star of $\mcP$, then {\sc Pull-by-$(k+1)^+$} or {\sc Pull-by-$k$} is applicable on the internal star.
Either of the above contradicts the termination of the algorithm.
\end{proof}

By Lemma~\ref{lemma03}, any optimal star contains at most $k$ uncovered vertices by the computed solution $\mcP$, and if exactly $k$ then the center of the optimal star is uncovered.
An optimal star containing exactly $k-i+1$ uncovered vertices is called a Type-$i$ optimal star, for $i = 1, 2$;
an optimal star containing at most $k-2$ uncovered vertices is called a Type-$3$ optimal star.

Let $Apx = |V(\mcP)|$ denote the number of vertices covered by the computed solution $\mcP$.
Let $Opt_i$ ($Apx_i$, respectively) denote the total number of vertices (covered vertices by $\mcP$, respectively) in all the Type-$i$ optimal stars, for $i = 1, 2, 3$;
and let $Apx_4$ denote the total number of covered vertices that fall outside of any optimal star.
It follows that
\begin{equation}
\label{eq01}
Opt = Opt_1 + Opt_2 + Opt_3, \mbox{ } Apx = Apx_1 + Apx_2 + Apx_3 + Apx_4.
\end{equation}

When $k \ge 3$, for each $i = 1, 2, 3$, $k-i+1 \ge 1$;
$Opt_i - Apx_i$ is the total number of uncovered vertices by $\mcP$ in all the Type-$i$ optimal stars,
and thus there are at least $\frac {Opt_i - Apx_i}{k-i+1}$ such Type-$i$ optimal stars.
On the other hand, there are at most $\frac {Apx_i}i$ such Type-$i$ optimal stars.
Therefore, $\frac {Opt_i - Apx_i}{k-i+1} \le \frac {Apx_i}i$, i.e.,
\begin{equation}
\label{eq02}
Opt_i \leq \frac{k+1}{i} Apx_i, i = 1, 2, 3,
\end{equation}
where the inequality becomes an equality when $i = 1, 2$.
When $k = 2$ and $i = 3$, $k-i+1 = 0$ and thus $Opt_i = Apx_i$, i.e., Eq.~(\ref{eq02}) still holds.

A satellite of a Type-$1$ optimal star is called a {\em critical} vertex (or simply a $c$-vertex) if it is covered by $\mcP$.
There are a total of $Apx_1$ $c$-vertices, stated as the first half in the next lemma.

\begin{lemma}
\label{lemma04}
There are a total of $Apx_1$ $c$-vertices, and each $c$-vertex is a satellite of a distinct internal $k$-star in $\mcP$.
\end{lemma}
\begin{proof}
Recall the definition of a Type-$1$ optimal star centered at $u$, that exactly $k-1$ satellites and the center $u$ are uncovered by $\mcP$.
If a $c$-vertex in this optimal star is a satellite of an internal $(k+1)^+$-star in $\mcP$, then {\sc Pull-by-$(k+1)^+$} is applicable on this internal star;
if a $c$-vertex in this optimal star is the center of an internal $k^+$-star in $\mcP$, then {\sc Collect} is re-applicable on this vertex.
Therefore, each $c$-vertex in this optimal star is a satellite of an internal $k$-star.
Next, if two $c$-vertices in a Type-$1$ optimal star, or in two different Type-$1$ optimal stars, both are satellites of the same internal $k$-star,
then {\sc Pull-by-$k$} is applicable on this internal $k$-star.
This proves the second half of the lemma.
\end{proof}

Lemma~\ref{lemma04} states that a $c$-vertex maps injectively to an internal $k$-star.
Continue from the proof of Lemma~\ref{lemma04} where the $c$-vertex $v$ is a satellite of a Type-$1$ optimal star centered at $u$,
and exactly $k-1$ satellites and the center $u$ are uncovered by $\mcP$.
One thus sees (from the perspective of $\mcP$) that the internal $k$-star to which $v$ belongs (or maps) is a $k_{v\mbox{-}u}$-star.
We call all the other vertices in the internal $k_{v\mbox{-}u}$-star the {\em adjugates} of $v$, or simply the $a$-vertices associated with the $c$-vertex $v$.
This way, each $c$-vertex is associated with exactly $k$ $a$-vertices,
and thus by Lemma~\ref{lemma04} there are exactly $k \times Apx_1$ $a$-vertices in total.

We next examine where these $a$-vertices sit in the optimal $k^+$-star packing solution $\mcQ^*$.
We first do not distinguish whether $k \ge 3$ or $k = 2$, while treating a $(k-2)$-star as a single vertex (i.e., the center vertex attached with no satellite) if $k = 2$.
Then we pay special attention when $k = 2$, as some details become very different, for example, Lemma~\ref{lemma05} versus Lemma~\ref{lemma06};
we also add another {\sc Pull} operation to achieve a better performance guarantee.

\subsubsection{Case 1 where $k \ge 2$}
\begin{lemma}
\label{lemma05}
Each Type-$2$ optimal star contains at most one $a$-vertex.
\end{lemma}
\begin{proof}
Recall that a Type-$2$ optimal star contains exactly $k-1$ uncovered vertices by $\mcP$.
Consider a Type-$2$ optimal star and we distinguish whether its center vertex is covered by $\mcP$ or not.

In the first case, the center vertex $w$ is uncovered.
If the optimal star contains two $a$-vertices associated with the same $c$-vertex $v$ in the internal $k_{v\mbox{-}u}$-star,
then {\sc Pull-by-$k$} is applicable on the $k_{v\mbox{-}u}$-star to replace it with a $k$-star centered at $u$ and a $k^+$-star centered at $w$
(see for an illustration at the top of Figure~\ref{fig04}).
If the optimal star contains two $a$-vertices associated with two different $c$-vertices, $v$ and $y$,
which are in the internal $k_{v\mbox{-}u}$-star and $k_{y\mbox{-}x}$-star, respectively,
then {\sc Pull-by-$(k, k)$} is applicable on these two internal $k$-stars to replace them with a $k^+$-star centered at $w$ and
two $k$-stars centered at $u$ and $x$, respectively
(see for an illustration at the bottom of Figure~\ref{fig04}, and note that $u \ne x$ by Lemma~\ref{lemma04}).
Therefore, in this case the Type-$2$ optimal star contains at most one $a$-vertex.

\begin{figure}[ht]
\centering\scalebox{1.0}{
  \setlength{\unitlength}{1bp}%
  \begin{picture}(202.45, 110.14)(0,0)
  \put(0,0){\includegraphics{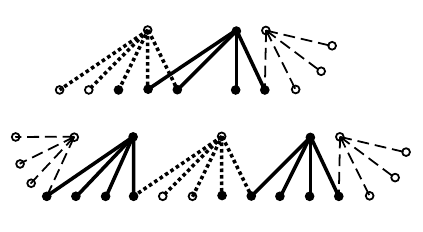}}
  \put(129.74,95.57){\fontsize{11.38}{13.66}\selectfont $u$}
  \put(72.37,95.58){\fontsize{11.38}{13.66}\selectfont $w$}
  \put(126.19,59.16){\fontsize{11.38}{13.66}\selectfont $v$}
  \put(165.25,44.52){\fontsize{11.38}{13.66}\selectfont $u$}
  \put(107.88,44.54){\fontsize{11.38}{13.66}\selectfont $w$}
  \put(161.70,8.12){\fontsize{11.38}{13.66}\selectfont $v$}
  \put(20.12,8.34){\fontsize{11.38}{13.66}\selectfont $y$}
  \put(35.42,46.82){\fontsize{11.38}{13.66}\selectfont $x$}
  \end{picture}%
}
\captionsetup{width=.99\textwidth}
\caption{An illustration to apply operation {\sc Pull-by-$k$} when the center vertex of the Type-$2$ optimal star is uncovered and contains two $a$-vertices.
	In this case, $k = 4$;
	at the top these two $a$-vertices are associated with the same $c$-vertex;
	at the bottom these two $a$-vertices are associated with different $c$-vertices.
	The edges in Type-$2$ optimal star are dotted.\label{fig04}}
\end{figure}
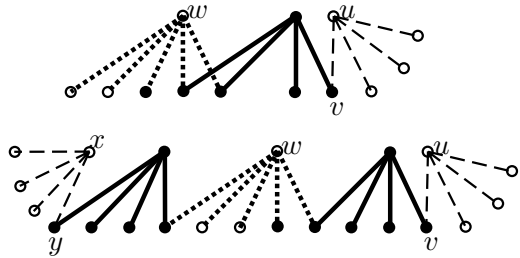

In the second case, the center vertex $w$ is covered and thus it is adjacent to the $k-1$ uncovered vertices in the optimal star.
This way, $w$ is a satellite of an internal $k^+_w$-star.
If $w$ is an $a$-vertex associated with a $c$-vertex $v$ in the internal $k_{v\mbox{-}u}$-star,
then {\sc Pull-by-$k$} is applicable on the $k_{v\mbox{-}u}$-star to replace it with two $k$-stars centered at $u$ and $w$, respectively.
If a satellite $y$ of the Type-$2$ optimal star is an $a$-vertex associated with a $c$-vertex $v$ in the $k_{v\mbox{-}u}$-star,
then depending on the size of the internal $k^+_w$-star,
either {\sc Pull-by-$(k, (k+1)^+)$} is applicable on the internal $k_{v\mbox{-}u}$-star and $(k+1)^+_w$-star to replace the $k_{v\mbox{-}u}$-star and $w$
	with a $k$-star centered at $u$ and a $k^+$-star centered at $w$,
or {\sc Pull-by-$(k, k)$} is applicable on the internal $k_{v\mbox{-}u}$-star and $k_w$-star to replace them with a $k$-star centered at $u$ and a $(k+1)^+$-star centered at $w$.
Therefore, in this case the Type-$2$ optimal star contains no $a$-vertex.
\end{proof}

\begin{theorem}
\label{thm01}
The algorithm {\sc LocalSearch-$k^+$} is an $O(|V|^3 |E|)$-time $(1 + \frac {k^2}{2k+1})$-approximation for the $k^+$-star packing problem for any $k \ge 2$.
\end{theorem}
\begin{proof}
We first analyze the time complexity of the algorithm.
Note that a {\sc Collect} operation needs to scan for the degrees of all the vertices in the remainder graph, which can be obtained in $O(|E|)$ time
(noting that by maintaining certain data structure such as a heap, this can be done faster),
followed by extracting a $k^+$-star and updating the remainder graph in $O(|V|)$ time.
A {\sc Pull} operation requires first trying one vertex out of a $(k+1)^+$-star, or trying two $k$-stars,
returning a maximum number $2(k+1)$ of vertices back to the remainder graph,
then extracting a maximum number $3$ of $k^+$-stars, and thus is done in $O(|V|^2 |E|)$ time.
Since each iteration increases the number of covered vertices by at least one, there are $O(|V|)$ iterations.
At the end of each operation, the algorithm re-applies {\sc Collect} operations on the centers of internal $k^+$-stars in $O(|E|)$ time.
These together imply the overall running time of the algorithm in $O(|V|^3 |E|)$.

To prove the performance ratio,
we recall that a Type-$2$ optimal star contains exactly $k-1$ uncovered vertices by $\mcP$, and thus at least two covered vertices.
Lemma~\ref{lemma05} states that at most one of the covered vertices in a Type-$2$ optimal star is an $a$-vertex.
It follows from Lemma~\ref{lemma04} that the total number of $a$-vertices is $k \times Apx_1 \le \frac 12 Apx_2 + Apx_3 + Apx_4$.
Combining this upper bound with Eqs.~(\ref{eq01},\ref{eq02}), and setting $\alpha = 1 + \frac {k^2}{2k+1}$, for any $k \ge 2$ we have
\[
\begin{array}{rcl}
Opt & = & Opt_1 + Opt_2 + Opt_3 \leq (k+1) Apx_1 + \frac {k+1}2 Apx_2 + \frac {k+1}3 Apx_3\\
   & = & \alpha Apx_1 + (k+1 - \alpha) Apx_1 + \frac {k+1}2 Apx_2 + \frac {k+1}3 Apx_3\\
   &\le & \alpha Apx_1 + (k+1 - \alpha) \cdot \frac 1k \cdot (\frac 12 Apx_2 + Apx_3 + Apx_4) + \frac {k+1}2 Apx_2 + \frac {k+1}3 Apx_3\\
   & = &  \alpha Apx_1 + (\frac {k+1 - \alpha}{2k} + \frac{k+1}2) Apx_2 + (\frac {k+1 - \alpha}k + \frac {k+1}3) Apx_3 + \frac {k+1 - \alpha}k Apx_4\\
   &\le & \max\{\alpha, \frac {k+1 - \alpha}{2k} + \frac {k+1}2, \frac {k+1 - \alpha}k + \frac {k+1}3\}(Apx_1 + Apx_2 + Apx_3 + Apx_4)\\
   &\le & (1 + \frac {k^2}{2k+1}) Apx.
\end{array}
\]
This proves the theorem.
\end{proof}

\subsubsection{Case 2 where $k = 2$}
The algorithm {\sc LocalSearch-$k^+$} works for every $k \ge 2$, and it is referred to as the general algorithm.
One sees that when $k = 2$, the general algorithm works out to be a $\frac 95$-approximation algorithm,
which ties the $\frac 95$-approximation algorithm by Huang et al.~\cite{HZG24} designed specifically for the $2^+$-star packing problem.
We note that the main design idea in the algorithm by Huang et al. is to compute a maximal number of balanced trees in the input graph and then to acquire $2^+$stars from the trees.

For convenience, we continue to use $k$ inside the name of the algorithm {\sc LocalSearch-$k^+$} and most notations as in the above,
but keep in mind that in this subsection $k = 2$.
We design an additional {\sc Pull} operation that works on three internal $k$-stars to increase the covered vertices.

\begin{definition}
\label{def06}
{\em (Operation {\sc Pull-by-$(k, k, k)$})}
Given three internal $k$-stars,
the operation removes all of them from $\mcP$ and applies {\sc Collect} operations to extract three stars, of which two are $k$-stars and the other is a $(k+1)^+$-star.
\end{definition}

\begin{lemma}
\label{lemma06}
Each Type-$2$ optimal star contains no $a$-vertex.
\end{lemma}
\begin{proof}
Recall that a Type-$2$ optimal star contains exactly one uncovered vertices by $\mcP$.
Consider a Type-$2$ optimal star and we distinguish whether its center vertex is covered by $\mcP$ or not.

In the first case, the center vertex $x$ is uncovered.
If the optimal star contains an $a$-vertex $w$ associated with a $c$-vertex $v$ in the internal $k_{v\mbox{-}u}$-star,
then {\sc Pull-by-$k$} is applicable on the $k_{v\mbox{-}u}$-star to replace it with two $k$-stars centered at $w$ and $u$, respectively
(see Figure~\ref{fig05} for an illustration).
Therefore, in this case the Type-$2$ optimal star contains no $a$-vertex.

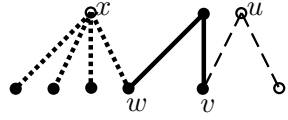
\begin{figure}[ht]
\centering\scalebox{1.0}{
  \setlength{\unitlength}{1bp}%
  \begin{picture}(117.34, 59.10)(0,0)
  \put(0,0){\includegraphics{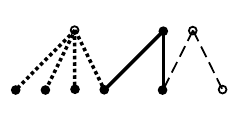}}
  \put(49.21,8.34){\fontsize{11.38}{13.66}\selectfont $w$}
  \put(94.66,44.52){\fontsize{11.38}{13.66}\selectfont $u$}
  \put(37.29,44.54){\fontsize{11.38}{13.66}\selectfont $x$}
  \put(77.12,8.12){\fontsize{11.38}{13.66}\selectfont $v$}
  \end{picture}%
}
\captionsetup{width=.99\textwidth}
\caption{An illustration to apply operation {\sc Pull-by-$k$} when the center vertex of the Type-$2$ optimal star is uncovered and contains an $a$-vertex.
	The edges in Type-$2$ optimal star are dotted.\label{fig05}}
\end{figure}

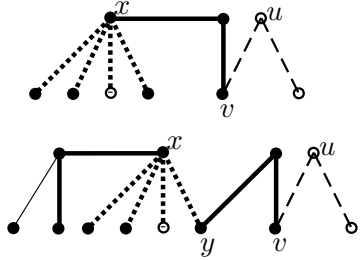
\begin{figure}[ht]
\centering\scalebox{1.0}{
  \setlength{\unitlength}{1bp}%
  \begin{picture}(145.45, 112.17)(0,0)
  \put(0,0){\includegraphics{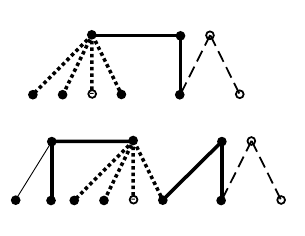}}
  \put(102.91,95.28){\fontsize{11.38}{13.66}\selectfont $u$}
  \put(45.32,97.61){\fontsize{11.38}{13.66}\selectfont $x$}
  \put(85.37,58.87){\fontsize{11.38}{13.66}\selectfont $v$}
  \put(122.77,44.52){\fontsize{11.38}{13.66}\selectfont $u$}
  \put(65.18,46.85){\fontsize{11.38}{13.66}\selectfont $x$}
  \put(105.23,8.12){\fontsize{11.38}{13.66}\selectfont $v$}
  \put(77.87,8.35){\fontsize{11.38}{13.66}\selectfont $y$}
  \end{picture}%
}
\captionsetup{width=.99\textwidth}
\caption{An illustration to apply operation {\sc Pull-by-$k$} when the center vertex of the Type-$2$ optimal star is an $a$-vertex (top),
	or to apply operation {\sc Pull-by-$(k, (k+1)^+)$} or {\sc Pull-by-$(k, k)$} when the center vertex of the Type-$2$ optimal star is covered but not an $a$-vertex (bottom).
	The edges in Type-$2$ optimal star are dotted;
	the thin edge in the bottom is a possible edge in the internal $k^+_x$-star, of which the existence determines which one of the two operations to be applied.\label{fig06}}
\end{figure}

In the second case, the center vertex $x$ is covered and thus it is adjacent to the uncovered vertex in the optimal star.
This way, $x$ is a satellite of an internal $k^+_x$-star.
If $x$ is an $a$-vertex associated with a $c$-vertex $v$ in the internal $k_{v\mbox{-}u}$-star,
then {\sc Pull-by-$k$} is applicable on the $k_{v\mbox{-}u}$-star to replace it with two $k$-stars centered at $u$ and $x$, respectively
(see the top of Figure~\ref{fig06} for an illustration).
If a satellite $y$ of the Type-$2$ optimal star is an $a$-vertex associated with a $c$-vertex $v$ in the $k_{v\mbox{-}u}$-star,
then depending on the size of the internal $k^+_x$-star,
either {\sc Pull-by-$(k, (k+1)^+)$} is applicable on the internal $k_{v\mbox{-}u}$-star and $(k+1)^+_x$-star to replace the $k_{v\mbox{-}u}$-star and $x$
	with a $k$-star centered at $u$ and a $k^+$-star centered at $x$,
or {\sc Pull-by-$(k, k)$} is applicable on the internal $k_{v\mbox{-}u}$-star and $k_x$-star to replace them with a $k$-star centered at $u$ and a $(k+1)^+$-star centered at $x$
(see the bottom of Figure~\ref{fig06} for an illustration).
Therefore, in this case the Type-$2$ optimal star contains no $a$-vertex either.
\end{proof}

Recall that a Type-$1$ optimal star contains exactly two uncovered vertices by $\mcP$.
If it contains only one $c$-vertex, i.e., it is a $k$-star, then it is called a Type-$11$ optimal star;
the other Type-$1$ optimal stars each containing at least two $c$-vertices are called Type-$12$ optimal stars (which are $(k+1)^+$-stars).
Also recall that a Type-$3$ optimal star contains no uncovered vertex by $\mcP$.
If it contains at most two $a$-vertices, then it is called a Type-$31$ optimal star;
the other Type-$3$ optimal stars each containing at least three $a$-vertices are called Type-$32$ optimal stars.
Correspondingly, let $Opt_{ij}$ ($Apx_{ij}$, respectively) denote the total number of vertices (covered vertices by $\mcP$, respectively) in all the Type-$ij$ optimal stars,
for $i = 1, 3$ and $j = 1, 2$.
It follows that, besides Eqs.~(\ref{eq01},\ref{eq02}),
\begin{eqnarray}
&~& Opt_i = Opt_{i1} + Opt_{i2}, \mbox{ } Apx_i = Apx_{i1} + Apx_{i2}, \mbox{ for } i = 1, 3;\label{eq03}\\
&~& Opt_{11} = 3 Apx_{11}, \mbox{ } Opt_{12} \le 2 Apx_{12}, \mbox{ and } Opt_{3j} = Apx_{3j}, \mbox{ for } j = 1, 2.\label{eq04}
\end{eqnarray}

\begin{lemma}
\label{lemma07}
Each Type-$32$ optimal star contains no $a$-vertex associated with a $c$-vertex in any Type-$11$ optimal star.
\end{lemma}
\begin{proof}
Recall that a Type-$32$ optimal star contains at least three $a$-vertices.

We first claim that its center $w$ cannot be an $a$-vertex.
Suppose to the contrary that $w$ and two satellites $z$ and $r$ in the Type-$32$ optimal star are $a$-vertices.
Since by Lemma~\ref{lemma04} a $c$-vertex is associated with exactly $k = 2$ $a$-vertices,
at least one of $z$ and $r$, say $z$, is not associated with the same $c$-vertex that $w$ is associated with.
Assume $z$ is associated with the $c$-vertex $y$ in an internal $k_{y\mbox{-}x}$-star and $w$ is associated with the $c$-vertex $v$ in an internal $k_{v\mbox{-}u}$-star.
Then {\sc Pull-by-$(k, k)$} is applicable on these two $k_{y\mbox{-}x}$-star and $k_{v\mbox{-}u}$-star to replace them
by three $k$-stars centered at $u, x$ and $z$ (or $w$), respectively, when $u \ne x$ (see Figure~\ref{fig07} for an illustration),
or by a $(k+1)$-star centered at $u$ and a $k$-star centered at $z$ (or $w$) when $u = x$.

\begin{figure}[ht]
\centering\scalebox{1.0}{
  \setlength{\unitlength}{1bp}%
  \begin{picture}(159.15, 61.41)(0,0)
  \put(0,0){\includegraphics{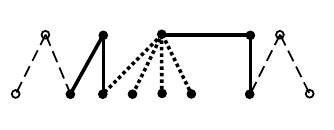}}
  \put(136.47,44.52){\fontsize{11.38}{13.66}\selectfont $u$}
  \put(78.89,46.85){\fontsize{11.38}{13.66}\selectfont $w$}
  \put(118.93,8.12){\fontsize{11.38}{13.66}\selectfont $v$}
  \put(76.56,8.34){\fontsize{11.38}{13.66}\selectfont $r$}
  \put(48.32,8.23){\fontsize{11.38}{13.66}\selectfont $z$}
  \put(23.93,44.52){\fontsize{11.38}{13.66}\selectfont $x$}
  \put(32.65,8.56){\fontsize{11.38}{13.66}\selectfont $y$}
  \end{picture}%
}
\captionsetup{width=.95\textwidth}
\caption{An illustration to apply operation {\sc Pull-by-$(k, k)$} when the center vertex of the Type-$32$ optimal star is an $a$-vertex.\label{fig07}}
\end{figure}
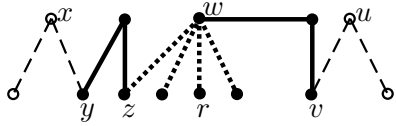

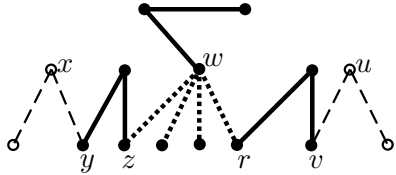
\begin{figure}[ht]
\centering\scalebox{1.0}{
  \setlength{\unitlength}{1bp}%
  \begin{picture}(159.15, 75.63)(0,0)
  \put(0,0){\includegraphics{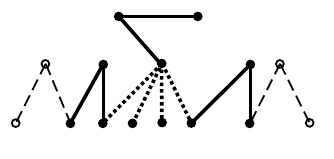}}
  \put(136.47,44.52){\fontsize{11.38}{13.66}\selectfont $u$}
  \put(78.89,46.85){\fontsize{11.38}{13.66}\selectfont $w$}
  \put(118.93,8.12){\fontsize{11.38}{13.66}\selectfont $v$}
  \put(92.01,8.34){\fontsize{11.38}{13.66}\selectfont $r$}
  \put(48.32,8.23){\fontsize{11.38}{13.66}\selectfont $z$}
  \put(23.93,44.52){\fontsize{11.38}{13.66}\selectfont $x$}
  \put(32.65,8.56){\fontsize{11.38}{13.66}\selectfont $y$}
  \end{picture}%
}
\captionsetup{width=.95\textwidth}
\caption{An illustration to apply operation {\sc Pull-by-$(k, k, k)$} when the center vertex of the Type-$32$ optimal star is a satellite in an internal $k$-star.\label{fig08}}
\end{figure}

Hence at least three satellites in the Type-$32$ optimal star are $a$-vertices, and we assume that $z$ and $r$ are associated with two different $c$-vertices
$y$ in an internal $k_{y\mbox{-}x}$-star and $v$ in an internal $k_{v\mbox{-}u}$-star, respectively.
When $w$ is in an internal $(k+1)^+$-star,
{\sc Pull-by-$(k, (k+1)^+)$} is applicable on these two $(k+1)^+$-star and $k_{v\mbox{-}u}$-star to strictly increase the number of covered vertices;
when $w$ is the center of an internal $k$-star,
{\sc Pull-by-$k$} is applicable on the $k_{v\mbox{-}u}$-star to strictly increase the number of covered vertices.
Therefore, $w$ is a satellite in an internal $k$-star.
If $u \ne x$, then {\sc Pull-by-$(k, k, k)$} is applicable on the three $k_{v\mbox{-}u}$-star, $k_{y\mbox{-}x}$-star and the $k$-star containing $w$ to replace them with
two $k$-stars centered at $u$ and $x$, respectively, and a $(k+1)^+$-star centered at $w$ (see for an illustration in Figure~\ref{fig08}).
In the remaining case where $u$ and $x$ collide,
Lemma~\ref{lemma04} states that the corresponding $c$-vertices $v$ and $y$ are in the same Type-$1$ optimal star.
Since $v$ and $y$ are distinct, this Type-$1$ optimal star is of Type-$12$.
This proves the lemma.
\end{proof}

\begin{theorem}
\label{thm02}
The algorithm {\sc LocalSearch-$2^+$} is an $O(|V|^4 |E|)$-time $\frac 32$-approxim\-ation for the $2^+$-star packing problem.
\end{theorem}
\begin{proof}
Note that we have an additional operation {\sc Pull-by-$(k, k, k)$} that works on three internal $k$-stars, 
and thus requires trying three $k$-stars,
returning a maximum number $3(k+1)$ of vertices back to the remainder graph,
then extracting a maximum number $3$ of $k^+$-stars, and thus is done in $O(|V|^3 |E|)$ time.
That is, the overall time complexity increases by an order, which is $O(|V|^4 |E|)$.

Since each Type-$31$ optimal star contains at most two $a$-vertices, at most $\frac 23$ of vertices of Type-$31$ optimal stars are $a$-vertices.
Among these $a$-vertices, assume $p_j$ of them are associated with the $c$-vertices in the Type-$1j$ optimal stars, for $j = 1, 2$.
That is,
\[
p_1 + p_2 \le \frac 23 Opt_{31}.
\]
Among those $Apx_4$ covered vertices by $\mcP$ but falling outside of the optimal solution $\mcQ^*$,
assume there are $q_j$ $a$-vertices which are associated with the $c$-vertices in the Type-$1j$ optimal stars, for $j = 1, 2$.
That is,
\[
q_1 + q_2 \le Apx_4.
\]
Using Lemma~\ref{lemma06} and~\ref{lemma07}, the total number of $a$-vertices associated with those $c$-vertices in the Type-$1j$ optimal stars, for $j = 1, 2$, is
\[
2 Apx_{11} \le p_1 + q_1, \mbox{ and } 2 Apx_{12} \le p_2 + q_2 + Opt_{32} \le \frac 23 Opt_{31} - p_1 + q_2 + Opt_{32},
\]
respectively.
Combining the above with Eqs.~(\ref{eq01}--\ref{eq04}), we obtain
\[
\begin{array}{rcl}
Opt &=& Opt_{11} + Opt_{12} + Opt_2 + Opt_{3} \le 3 Apx_{11} + 2 Apx_{12} + \frac 32 Apx_2 + Apx_3\\
   &=& \frac 32 Apx_{11} + \frac 12 Apx_{12} + \frac 32 (Apx_1 + Apx_2) + Apx_3\\
   &\le & \frac 34 (p_1 + q_1) + \frac 14 (\frac 23 Opt_{31} - p_1 + q_2 + Opt_{32}) + \frac 32 (Apx_1 + Apx_2) + Apx_3\\
   &=& \frac 12 p_1 + \frac 34 q_1 + \frac 14 q_2 + \frac 16 Opt_{31} + \frac 14 Opt_{32} + \frac 32 (Apx_1 + Apx_2) + Apx_3\\
   &\le & \frac 13 Opt_{31} + \frac 34 Apx_4 + \frac 16 Opt_{31} + \frac 14 Opt_{32} + \frac 32 (Apx_1 + Apx_2) + Apx_3\\
   &\le & \frac 12 Opt_3 + \frac 34 Apx_4 + \frac 32 (Apx_1 + Apx_2) + Apx_3\\
   &\le & \frac 32 (Apx_1 + Apx_2 + Apx_3 + Apx_4)
	= \frac 32 Apx.
\end{array}
\]
This finishes the proof.
\end{proof}

\section{Approximating the $k^-/t$-star packing problem}
Unlike the NP-hard $k^+$-star packing problem we consider in the last section,
the star packing problem using vertex-disjoint $k^-$-stars, i.e., stars with $1$ up to $k$ satellites, to cover the maximum number of vertices,
is the {\em sequential} $k^-$-star packing problem and solvable in $O(\sqrt{|V|} |E|)$ time for any $k \ge 1$~\cite{HK86,BG11}.

Let $\mcQ_0$ denote an optimal $k^-$-star packing for an input graph $G = (V, E)$, and let $V_0 = V \setminus V(\mcQ_0)$ denote the set of uncovered vertices by $\mcQ_0$.
One sees that each vertex $v \in V_0$ can be adjacent to only the centers of the internal $k$-stars in $\mcQ_0$,
since otherwise $v$ would be covered to achieve a better packing.
These internal $k$-stars in $\mcQ_0$ of which the center is adjacent to a vertex in $V_0$ are called {\em critical}.
For the same reason, each satellite of a critical $k$-star can be adjacent to only the centers of the internal $k$-stars in $\mcQ_0$,
and recursively, those internal $k$-stars in $\mcQ_0$ of which the center is adjacent to a satellite of critical $k$-star become {\em critical} too.
We denote by $\mcK(v) \subseteq \mcQ_0$ the collection of all critical $k$-stars associated with the vertex $v \in V_0$,
and $\mcK(V_0) = \cup_{v \in V_0} \mcK(v)$.

Note that $v \in V_0$ and all the satellites covered by $\mcK(v)$ together form an independent set in the input graph $G$,
and their neighbors are the centers of the critical $k$-stars in $\mcK(v)$.

Given $k > t \ge 1$, in the $k^-/t$-star packing problem one uses vertex-disjoint $k^-$-stars except $t$-stars to cover the maximum number of vertices.
The problem is NP-hard for any $k > t \ge 1$~\cite{HK86},
and admits a straightforward $(1 + \frac 1t)$-approximation algorithm, for $t \ge 2$,
by first computing the $\mcQ_0$ and then removing one satellite from each $t$-star in $\mcQ_0$~\cite{LL21}.
Below we continue to use the idea in this simple approximation, but we seek for a better $\mcQ_0$ in which there are relatively fewer $t$-stars than the other stars.
For instance, in the extreme case where $\mcQ_0$ contains no $t$-star, it is an optimal solution to the $k^-/t$-star packing problem.

So we consider $k > t \ge 2$,
fix an optimal $k^-/t$-star packing $\mcQ^*$ for discussion and let $V^*_0 = V \setminus V(\mcQ^*)$ denote the set of uncovered vertices by $\mcQ^*$.
We observe the optimal packing in the next lemma.

\begin{lemma}
\label{lemma08}
There exists an optimal sequential $k^-$-star packing $\mcQ_0$ for the input graph in which the uncovered vertex set $V_0 \subseteq V^*_0$.
\end{lemma}
\begin{proof}
The proof is constructive, in that we start with an optimal sequential $k^-$-star packing $\mcQ_0$ for the input graph,
and if there is a vertex $v \in V_0 \setminus V^*_0$ then we swap $v$ with a vertex $u \in V^*_0 \setminus V_0$ uncovered by $\mcQ_0$.
The stars in $\mcQ^*$ are referred to as {\em optimal stars}, and assume $v$ is in the optimal star $S^*$.

\begin{figure}[ht]
\centering\scalebox{1.0}{
  \setlength{\unitlength}{1bp}%
  \begin{picture}(206.80, 62.25)(0,0)
  \put(0,0){\includegraphics{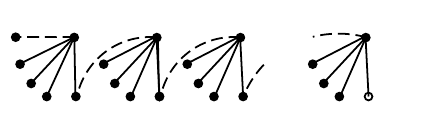}}
  \put(7.63,47.70){\fontsize{11.38}{13.66}\selectfont $v$}
  \put(34.83,8.12){\fontsize{11.38}{13.66}\selectfont $v_1$}
  \put(34.42,47.37){\fontsize{11.38}{13.66}\selectfont $c_1$}
  \put(62.29,47.48){\fontsize{11.38}{13.66}\selectfont $c_2$}
  \put(74.97,8.12){\fontsize{11.38}{13.66}\selectfont $v_2$}
  \put(102.42,47.48){\fontsize{11.38}{13.66}\selectfont $c_3$}
  \put(175.30,8.12){\fontsize{11.38}{13.66}\selectfont $u = v_i$}
  \put(162.62,47.48){\fontsize{11.38}{13.66}\selectfont $c_i$}
  \put(115.10,8.12){\fontsize{11.38}{13.66}\selectfont $v_3$}
  \put(124.16,35.78){\fontsize{11.38}{13.66}\selectfont $\ldots$}
  \end{picture}%
}
\captionsetup{width=.95\textwidth}
\caption{An illustration to swap $v \in V_0 \setminus V^*_0$ with a vertex $u \in V^*_0 \setminus V_0$ uncovered by $\mcQ_0$.
	Here $k = 4$, the solid edges are in $\mcQ_0$ and the dashed edges are in $\mcQ^*$.\label{fig09}}
\end{figure}
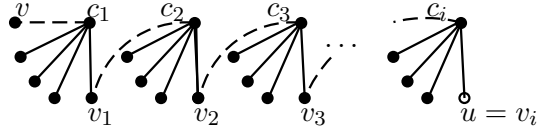

Note from the above argument that all the neighbors of $v$ are the centers of some critical $k$-stars in $\mcK(v)$,
and one neighbor in $\mcQ^*$ is denoted as $c_1$ and assume $c_1$ is the center of a critical $k$-star $S_1 \in \mcK(v)$
(see for an illustration in Figure~\ref{fig09}).
It follows that there exists a satellite of $S_1$, denoted as $v_1$, which is not covered by the optimal star $S^*$.
If $v_1$ is not covered by any other optimal star, that is, $v_1 \in V^*_0 \setminus V_0$, then we have found $u = v_1$.
Otherwise, $v_1$ is in an optimal star $S^*_1$ other than $S^*$.
Next, all the neighbors of $v_1$ are the centers of some critical $k$-stars in $\mcK(v)$,
and one neighbor in $\mcQ^*$ is denoted as $c_2$ and assume $c_2$ is the center of the critical $k$-star $S_2 \in \mcK(v)$
(see for an illustration in Figure~\ref{fig09}).
Then there exists a satellite $v_2$ of $S_1$ or $S_2$ which is not covered by the optimal stars $S^*$ and $S^*_1$.
Again, if $v_2$ is not covered by any other optimal star, then we have found $u = v_2$.
Otherwise, $v_2$ is in an optimal star $S^*_2$ other than $S^*$ and $S^*_1$.
This process is repeated, and by the finiteness of $\mcQ^*$ the process terminates with the vertex $u$ found.
Afterwards, we swap $v$ with $u$ and update the involved critical $k$-stars correspondingly, that is, $v_1$ is replaced by $v$,
$v_2$ is replaced by $v_1$, etc.
\end{proof}

The above constructive proof of Lemma~\ref{lemma08} tells that,
given an optimal $k^-/t$-star packing $\mcQ^*$ and an optimal sequential $k^-$-star packing $\mcQ_0$,
how to modify those critical $k$-stars in $\mcK(V_0)$ into one such that $V_0 \subseteq V^*_0$.
Note that $\mcQ_0$ can be computed in polynomial time and below we assume without loss of generality that for the optimal $k^-/t$-star packing $\mcQ^*$, $V_0 \subseteq V^*_0$.

In our algorithm {\sc LocalSearch-$k^-/t$} for the $k^-/t$-star packing problem (see for a high-level description in Figure~\ref{fig10}), where $k > t \ge 2$,
we propose local search operations to work with those stars in $\mcQ_0 \setminus \mcK(V_0)$ to reduce the number of $t$-stars,
or to keep the number of $t$-stars while increasing the number of other stars.
This way, only $O(|V|)$ iterations are done and at the end one satellite of each remaining $t$-star in $\mcQ_0$ is discarded to become a solution $\mcP$
for the $k^-/t$-star packing problem.

\subsection{The algorithm}
The algorithm works with $\mcQ_0$, which maintains to be an optimal $k^-$-star packing and is referred to as the {\em current} $k^-$-star packing.
Below we describe local search operations at the presence of at least a $t$-star inside $\mcQ_0$, as otherwise $\mcQ_0$ is an optimal $k^-/t$-star packing.
These operations are collectively called {\sc Revise}, but they separately work on different combinations of stars in $\mcQ_0$ of which at least one is a $t$-star.
An operation is {\em applicable} means it reduces the number of $t$-stars while keeping the same vertices covered,
or it keeps the number of $t$-stars while increasing the number of other stars.

\begin{definition}
\label{def07}
{\em (Operation {\sc Revise-$t$})}
Given a $t$-star, the operation revises it into two stars, of which one is a $(t-2)$-star and the other is a $1$-star.
\end{definition}

Note that an operation {\sc Revise-$t$} applies to a $t$-star if two satellites of the star are adjacent in the input graph $G$
(and there is no other possibility since $\mcQ_0$ is an optimal $k^-$-star packing),
and thus the operation takes $O(1)$ time.
Specially, when $t = 2$, no {\sc Revise-$t$} operation is applicable,
since we do not want to form the two satellites into a $1$-star while leaving out the center uncovered.

\begin{definition}
\label{def08}
{\em (Operation {\sc Revise-$(t, i)$})}
Given a $t$-star and an $i$-star, the operation revises them into two $k^-$-stars of which none is a $t$-star, or into three $k^-$-stars of which at most one is a $t$-star.
\end{definition}

We remark that there are multiple scenarios for {\sc Revise-$(t, i)$} to be applicable,
for example, when a satellite of the $t$-star is adjacent to the center of the $i$-star and $i \ne t-1, k$,
or when two satellites of the $t$-star are both adjacent to the center of the $i$-star and $i = t-1$,
or when the center of the $t$-star is adjacent to a satellite of the $i$-star and $i \ne 1, t+1$.
Nevertheless, since at most $t + k + 2$ vertices are involved, the operation takes $O(1)$ time.

\begin{definition}
\label{def09}
{\em (Operation {\sc Revise-$(t, t, i)$})}
Given two $t$-stars and an $i$-star with $i \in \{1, t-1, t+1\}$, the operation revises them into two or three $k^-$-stars of which none is a $t$-star.
\end{definition}

Again, we remark that there are multiple scenarios for {\sc Revise-$(t, t, i)$} to be applicable,
for example, when the centers of the two $t$-stars are adjacent to different vertices of the $1$-star, respectively,
then the $1$-star can be split so that two $(t+1)$-stars are formed.
Since at most $3t + 4$ vertices are involved, the operation takes $O(1)$ time.

\begin{lemma}
\label{lemma09}
When none of the {\sc Revise} operation is applicable, the optimal sequential $k^-$-star packing $\mcQ_0$ satisfies the following:
\begin{description}
\parskip=0pt
\item[(1)]
	$V_0 = V \setminus V(\mcQ_0)$ and $\mcK(V_0)$ remain unchanged throughout the algorithm.
\item[(2)]
	Each satellite of every $t$-star is adjacent to only the centers of $(t-1)$-stars or $k$-stars.
\item[(3)]
	The center of every $t$-star is adjacent to only the satellites of $(t+1)$-stars or the centers of the other stars.
\end{description}
\end{lemma}
\begin{proof}
One sees that no operation touches any uncovered vertex or produces any new uncovered vertex, and thus $V_0$ remains unchanged.
A $k$-star of $\mcK(V_0)$ cannot be involved in an operation {\sc Revise-$(t, k)$},
since all its $k$ satellites are adjacent to its center only in the induced subgraph of $t + k + 2$ vertices and thus the $k$-star remains unchanged in the operation.
This proves statement (1).

To prove statement (2), firstly, two satellites of a $t$-star cannot be adjacent, as otherwise {\sc Revise-$t$} is applicable.
Secondly, if a satellite of the $t$-star is adjacent to a satellite of another $i$-star,
then when $i \ge 2$, {\sc Revise-$(t, i)$} can be applied to form a $(t-1)$-star, an $(i-1)$-star, and a $1$-star containing those two satellites;
when $i = 1$, {\sc Revise-$(t, i)$} can be applied to form a $(t-1)$-star and a $2$-star if $t \ge 3$,
or otherwise the roles of center and satellite in the $1$-star are swapped
(noting that if there is a satellite of another $t$-star adjacent to the center of the $1$-star,
then {\sc Revise-$(t, t, 1)$} can be applied to form a $(t+1)$-star and two $1$-stars).
Lastly, if a satellite of the $t$-star is adjacent to the center of another $i$-star with $i \ne t-1, k$,
then {\sc Revise-$(t, i)$} can be applied to form a $(t-1)$-star and an $(i+1)$-star.
This proves statement (2), and statement (3) can be similarly observed.
\end{proof}

Using all the above defined {\sc Revise} operations, our algorithm {\sc LocalSearch-$k^-/t$} applies one of them in each iteration;
when none of them is applicable, the algorithm removes one satellite from each $t$-star in the current packing $\mcQ_0$ and
returns it as the solution, denoted as $\mcP$, to the $k^-/t$-star packing problem.
A high-level description of the algorithm is depicted in Figure~\ref{fig10}.

\begin{figure}[h]
\begin{center}
\framebox{
\begin{minipage}{0.95\textwidth}
{\sc LocalSearch-$k^-/t$} for $k^-/t$-star packing when $k > t \ge 2$:\\
Input: A connected graph $G = (V, E)$;\\
Output: A collection $\mcP$ of vertex-disjoint $k^-/t$-stars in $G$.
\begin{enumerate}
\parskip=0pt
\item
	Compute an optimal $k^-$-star packing $\mcQ_0$, set $V_0 = V \setminus V(\mcQ_0)$ and $\mcK(V_0)$;
\item
	While (a {\sc Revise} operation is applicable) do
	\begin{enumerate}
	\parskip=0pt
	\item
		apply the operation to update $\mcQ_0$;
	\end{enumerate}
\item
	Remove one satellite from each $t$-star in $\mcQ_0$ to become the final packing $\mcP$.
\end{enumerate}
\end{minipage}}
\end{center}
\captionsetup{width=.99\textwidth} 
\caption{A high-level description of {\sc LocalSearch-$k^-/t$} for $k^-/t$-star packing.\label{fig10}}
\end{figure}

\subsection{Performance analysis}
We fix an optimal $k^-/t$-star packing $\mcQ^*$ for discussion,
and by Lemma~\ref{lemma08} we assume all the covered vertices by $\mcQ^*$ are covered by the optimal $k^-$-star packing $\mcQ_0$, i.e., $V(\mcQ^*) \subseteq V(\mcQ_0)$.
We examine how the vertices of $V(\mcQ_0)$ are covered by $\mcQ^*$ and $\mcP$, respectively,
where $\mcQ_0$ is the optimal $k^-$-star packing on which no {\sc Revise} operation is applicable.

To this purpose, we divide the $t$-stars in $\mcQ_0$ into three types (see for an illustration in Figure~\ref{fig11}):
\begin{itemize}
\parskip=0pt
\item
	In a Type-$1$ $t$-star, at least one vertex ($v_1$ in the illustration) is not covered by $\mcQ^*$.
\item
	In a Type-$2$ $t$-star, all its vertices are covered by $\mcQ^*$ but at least one edge ($(c_2, v_2)$ in the illustration) is not in $\mcQ^*$
	(i.e., one satellite of the $t$-star is connected to a vertex outside of the $t$-star in $\mcQ^*$).
\item
	In a Type-$3$ $t$-star, all its edges are in $\mcQ^*$ (i.e., the $t$-star is inside a $(t+1)^+$-star in $\mcQ^*$,
	or the center of the $t$-star is connected to at least one vertex, $v_3$ in the illustration, outside of the $t$-star in $\mcQ^*$).
\end{itemize}

\begin{figure}[ht]
\centering\scalebox{1.0}{
  \setlength{\unitlength}{1bp}%
  \begin{picture}(218.47, 61.59)(0,0)
  \put(0,0){\includegraphics{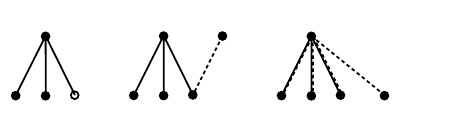}}
  \put(35.21,8.12){\fontsize{11.38}{13.66}\selectfont $v_1$}
  \put(20.65,46.92){\fontsize{11.38}{13.66}\selectfont $c_1$}
  \put(92.10,8.35){\fontsize{11.38}{13.66}\selectfont $v_2$}
  \put(77.32,47.03){\fontsize{11.38}{13.66}\selectfont $c_2$}
  \put(184.46,8.35){\fontsize{11.38}{13.66}\selectfont $v_3$}
  \put(148.24,46.92){\fontsize{11.38}{13.66}\selectfont $c_3$}
  \end{picture}%
}
\captionsetup{width=.95\textwidth}
\caption{An illustration of three types of $t$-stars in $\mcQ_0$, here $t = 3$ and the $t$-star centered at $c_i$ is Type-$i$,
	the solid edges are in $\mcQ_0$ and the dotted edges are in $\mcQ^*$,
	and the empty vertex is uncovered by $\mcQ^*$.\label{fig11}}
\end{figure}
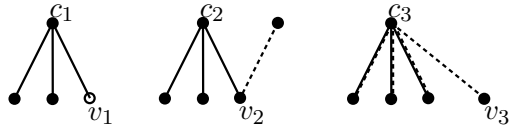

Below we partition the covered vertices by $\mcQ_0$ into vertex-disjoint {\em regions} such that each region contains one or two $t$-stars.
Or equivalently, we {\em map} a $t$-star or a $t$-star pair injectively to a collection of stars or partial stars in $\mcQ_0$,
by which to count the number of {\em distinct} vertices each $t$-star in $\mcQ_0$ is {\em associated} with.
This way, we obtain the average number of vertices that are kept in $\mcP$ when deleting a satellite of a $t$-star in $\mcQ_0$.
The partitioning or mapping uses the termination conditions of the algorithm {\sc LocalSearch-$k^-/t$}.
For $0 < x < 1$ and a star $S$, we denote by $x S$ a {\em partial} star of $S$ which is a set containing exactly $x |V(S)|$ vertices of $S$.
Any two partial stars of a star are {\em vertex-disjoint}, i.e., they do not share any common fractional vertex.

\paragraph*{The mapping:}
For ease of presentation, we use the illustration in Figure~\ref{fig11}.
Note that the stars below are all in $\mcQ_0$, while the other adjacencies are in $\mcQ^*$.
\begin{enumerate}
\parskip=0pt
\item
	For a Type-$3$ $t$-star $S$ as illustrated in Figure~\ref{fig11}, the center $c_3$ is connected to at least one vertex $v_3$ outside $S$ in $\mcQ^*$,
	which by Lemma~\ref{lemma09}(3) is either the center of another star or a satellite of a $(t+1)$-star in $\mcQ_0$.
	\begin{enumerate}
	\parskip=0pt
	\item
		If $v_3$ is the center of an $i$-star $S'$ with $i \ne t, t+1$, then $S$ maps to $S \cup S'$.
	\item
		If $v_3$ is the center of a Type-$1$ $t$-star $T$, then $S$ and $T$ map to $S \cup T$.
	\item
		If $v_3$ is in a $(t+1)$-star $S'$, then $S$ maps to $S \cup \frac 12 S'$.
	\item
		If $v_3$ is the center of a Type-$2$ $t$-star $T$ ($S$ and $T$ is referred to as a {\em $t$-star pair}),
		then in $\mcQ^*$ each satellite of $T$ is not connected to $v_3$ anymore, 
		but by Lemma~\ref{lemma09}(2) it is connected to the center of a $(t-1)$-star $S'$ or a $k$-star $T'$ in $\mcQ^*$.
		The $t$-star pair $S$ and $T$ maps to all such $(t-1)$-stars $S'$, if any, plus the partial $k$-star $\frac {x(T')}k T'$ of each such $k$-star $T'$, if any,
		i.e., $S \cup T \cup (\cup_{S'} S') \cup (\cup_{T'} \frac {x(T')}k T')$, where $x(T')$ is the number of satellites of $T$ connected to the center of $T'$ in $\mcQ^*$.
	\end{enumerate}
\item
	For a Type-$2$ $t$-star $S$ that is not involved in the first part, as illustrated in Figure~\ref{fig11},
	by Lemma~\ref{lemma09}(2) at least one satellite $v_2$ is connected to the center of a $(t-1)$-star $S'$ or a $k$-star $S''$ in $\mcQ^*$.
	Then $S$ maps to either $S \cup S'$ if $S'$ exists or $S \cup \frac 1k S''$ otherwise.
\item
	For a Type-$1$ $t$-star $S$ that is not involved in the first part, as illustrated in Figure~\ref{fig11}, $S$ maps to itself $S$.
\end{enumerate}

For each $t$-star $S$ or a $t$-star pair $(S, T)$, we denote the set of vertices it maps to by $U(S)$ or $U(S, T)$, respectively.
When $S$ or $(S, T)$ is clear from the context, we simplify it as $U$.
We will show that no other $t$-star maps to any vertex of $U$, i.e., the mapping is injective.
Next, for any subset $U \subseteq V(\mcQ_0)$, we define the ratio $r(U) = \frac {Opt(U)}{Apx(U)}$,
where $Opt(U)$ and $Apx(U)$ denote the number of vertices of $U$ that are covered by $\mcQ^*$ and $\mcP$, respectively.
One sees that if a $t$-star $S$ or a $t$-star pair $(S, T)$ maps to itself, then $r(U) \le 1$;
and for the set $U$ of vertices of $V(\mcQ_0)$ that are not mapped by any $t$-star or $t$-star pair, $r(U) \le 1$ too.
The next four lemmas consider the interesting cases (1a, 1c, 1d and 2) in the mapping process.

\begin{lemma}
\label{lemma10}
If a Type-$3$ $t$-star $S$ maps to $U = S \cup S'$ as in case 1a, then $r(U) \le 1 + \frac 1{t+2}$.
\end{lemma}
\begin{proof}
Again see for an illustration in Figure~\ref{fig11}, where $S$ is centered at $c_3$ and $c_3$ is connected to an outside vertex $v_3$ in $\mcQ^*$.
In case 1a, $v_3$ is the center of $S'$ which is an $i$-star with $i \ne t, t+1$.
Since $v_3$ is a satellite of the $(t+1)^+$-star centered at $c_3$ in $\mcQ^*$,
if $i > 1$ then no other Type-$3$ $t$-star maps to $S'$.
If $i = 1$ and there is another Type-$3$ $t$-star mapping to $S'$, then it maps through the vertex of $S'$ other than $v_3$, i.e., the satellite;
however, in such a case {\sc Revise-$(t, t, 1)$} is applicable, a contradiction.
The ratio $r(U) \le \max_{i \ne t, t+1} \frac {(t+1) + (i+1)}{t + (i+1)} = 1 + \frac 1{t+2}$, in which the maximum is achieved at $i = 1$.
\end{proof}

\begin{lemma}
\label{lemma11}
If a Type-$3$ $t$-star $S$ maps to $U = S \cup \frac 12 S'$ as in case 1c, then $r(U) \le 1 + \frac 1{1.5t + 1}$.
\end{lemma}
\begin{proof}
Again see for an illustration in Figure~\ref{fig11}, where $S$ is centered at $c_3$ and $c_3$ is connected to an outside vertex $v_3$ in $\mcQ^*$.
In case 1c, $v_3$ is in a $(t+1)$-star $S'$.
Again since $v_3$ is a satellite of the $(t+1)^+$-star centered at $c_3$ in $\mcQ^*$, if there are three Type-$3$ $t$-stars mapping to $S'$,
then two of them map through two different satellites (besides one mapping through the center) of $S'$;
however, in such a case {\sc Revise-$(t, t, t+1)$} is applicable, a contradiction.
This shows that the partial $(t+1)$-star $\frac 12 S'$ mapped by $S$ is vertex-disjoint with the other partial $(t+1)$-star $\frac 12 S'$.
The ratio $r(U) \le \frac {(t+1) + \frac 12 (t+2)}{t + \frac 12 (t+2)} = 1 + \frac 1{1.5t + 1}$.
\end{proof}

\begin{lemma}
\label{lemma12}
If a Type-$3$ $t$-star $S$ and a Type-$2$ $t$-star $T$ map together to $U = S \cup T \cup (\cup_{S'} S') \cup (\cup_{T'} \frac {x(T')}k T')$ as in case 1d,
then $r(U) \le 1 + \frac 1{t+1 + 1/k}$.
\end{lemma}
\begin{proof}
Firstly, if the $(t-1)$-star $S'$ exists, then no other $t$-star pair or a Type-$2$ $t$-star maps to it (in case 1d or in case 2, respectively)
since otherwise {\sc Revise-$(t, t, t-1)$} would be applicable, a contradiction.
No Type-$3$ $t$-star maps to $S'$ either by Lemma~\ref{lemma10}.
Next, consider a $k$-star $T'$ of which the center is connected to $x(T')$ satellites of $T$ in $\mcQ^*$.
When $k > t+1$ (i.e., excluding case 1c), since no $t$-star or $t$-star pair maps to a satellite of $T'$ and the center of $T'$ is connected to at most $k$ other vertices in $\mcQ^*$,
the partial $k$-star $x(T') T'$ mapped by the pair $(S, T)$ is vertex-disjoint with any other partial $k$-stars of $T'$ mapped by others.
When $k = t+1$, no Type-3 $t$-star maps to $T'$ through a satellite of $T'$ either since otherwise {\sc Revise-$(t, t, t+1)$} would be applicable.
This shows that every vertex of $U$ is mapped only to the $t$-star pair $(S, T)$.

Recall from Lemma~\ref{lemma09}(2) that each satellite of $T$ is connected to the center of a $(t-1)$-star or a $k$-star in $\mcQ^*$,
but no two satellites of $T$ can be connected to the center of the same $(t-1)$-star in $\mcQ^*$ since otherwise {\sc Revise-$(t, t-1)$} is applicable when $t \ge 3$
or {\sc Revise-$(t, t, 1)$} is applicable when $t = 2$.
Suppose there are $p$ satellites connecting to the centers of $p$ $(t-1)$-stars, respectively, and $t - p$ satellites connecting to the centers of $k$-stars in $\mcQ^*$.
If $p = 0$, then $q = t \ge 2$ and hence the ratio $r(U) \le \frac {2(t+1) + \frac tk (k+1)}{2t + \frac tk (k+1)} \le 1 + \frac 1{1.5t + 1/k} \le 1 + \frac 1{t+1 + 1/k}$;
if $p \ge 1$, 
then $r(U) \le \frac {2(t+1) + p t + \frac {(t-p)(k+1)}k}{2t + p t + \frac {(t-p)(k+1)}k} \le 1 + \frac 1{1.5t + 1} \le 1 + \frac 1{t+1 + 1/k}$ due to $t - p \ge 0$ and $k > t \ge 3$,
or $r(U) \le 1 + \frac 1{1.5t + 1/2} \le 1 + \frac 1{t+1 + 1/k}$ when $k > t = 2$.
\end{proof}

\begin{lemma}
\label{lemma13}
If a Type-$2$ $t$-star $S$ maps to $U = S \cup S'$ or $U = S \cup \frac 1k S''$ as in case 2, then $r(U) \le 1 + \frac 1{t+1 + 1/k}$.
\end{lemma}
\begin{proof}
The proof is similar to Lemma~\ref{lemma12}, but beginning with the Type-$2$ $t$-star $T$.
The ratio $r(U) \le \max\{\frac {(t+1) + t}{t + t}, \frac {(t+1) + \frac 1k (k+1)}{t + \frac 1k (k+1)}\} = 1 + \frac 1{t+1 + 1/k}$.
\end{proof}

\begin{theorem}
\label{thm03}
The algorithm {\sc LocalSearch-$k^-/t$} is an $O(|V|^4 |E|)$-time $(1 + \frac 1{t+1 + 1/k})$-approximation for the $k^-/t$-star packing problem for any $k > t \ge 2$.
\end{theorem}
\begin{proof}
Note that our algorithm calls the $O(\sqrt{|V|} |E|)$-time algorithm~\cite{BG11} to compute an optimal sequential $k^-$-star packing first,
followed by iteratively applying a {\sc Revise} operation which scans for a combination of at most three stars to reduce the number of $t$-stars,
or to keep the same number of $t$-stars but increasing the number of other stars.
Such an iteration can be done in $O(|V|^3 |E|)$ time.
Lastly, one satellite from each remaining $t$-star is removed to form a feasible $k^-/t$-star packing, which takes only $O(|V|)$ time.
The overall running time of the algorithm is thus $O(|V|^4 |E|)$.

The performance ratio is done in Lemmas~\ref{lemma10}--\ref{lemma13}.
\end{proof}

We remark that for the $k^-/t$-star packing problem where $k = \infty$ and $t \ge 2$, i.e., all but $t$-stars are candidates,
the same argument above by removing the star size constraint proves that the algorithm {\sc LocalSearch-$k^-/t$} is a $(1 + \frac 1{t+2})$-approximation for $t \ge 2$.

\section{Conclusion}
We studied two non-sequential star packing problems using large stars or small stars only, specially, the $k^+$-star packing for $k \ge 2$ and the $k^-/t$-star packing for $k > t \ge 2$,
from the perspective of approximation algorithms.
We designed a local search algorithm for each of them, which is $(1 + \frac {k^2}{2k + 1})$-approximation and $(1 + \frac 1{t + 1 + 1/k})$-approximation, respectively;
and furthermore by adding another local operation, the algorithm for $k^+$-star packing problem becomes a $\frac 32$-approximation for the $2^+$-star packing problem.
They improve the state-of-the-art $(1 + \frac k2)$-, $(1 + \frac 1t)$-, and $\frac 95$-approximation, respectively.
As pointed out by Li and Lin~\cite{LL21}, it would be interesting to design approximation algorithms for the $\mcH$-packing problem where $\mcH$ contains two stars only.
We also note that there seems no existing approximation algorithm for the $k^-/1$-star packing, i.e., $k > t = 1$,
unless when $k = 2$ the problem becomes the $P_3$-packing problem.


\end{document}